\documentclass[11pt]{article}

\usepackage[backend=bibtex8,style=alphabetic,maxbibnames=100,maxcitenames=3]{biblatex}
\usepackage[utf8]{inputenc}
\usepackage{fullpage}
\usepackage{hyperref}
\usepackage{amsfonts,amssymb,amsthm,bbm,mathtools}
\usepackage{thm-restate}
\usepackage{footnote}
\usepackage{authblk}
\usepackage{enumerate}
\usepackage{algorithm}
\usepackage{marginnote}
\usepackage{pgfplots}
\usepackage{wrapfig}
\usepackage{xcolor}

\usetikzlibrary{arrows,decorations.markings}

\bibliography{bibliography}

\hypersetup{
	colorlinks = true,
	citecolor = blue
}

\let\originalleft\left
\let\originalright\right
\renewcommand{\left}{\mathopen{}\mathclose\bgroup\originalleft}
\renewcommand{\right}{\aftergroup\egroup\originalright}

\newtheorem{theorem}{Theorem}

\newtheorem{corollary}[theorem]{Corollary}
\newtheorem{proposition}[theorem]{Proposition}
\theoremstyle{definition}

\DeclareMathOperator{\be}{H}
\DeclareMathOperator{\poly}{poly}

\title{Quantum Speedups for Exponential-Time \\ Dynamic Programming Algorithms\thanks{This work is supported by the ERC Advanced Grant MQC and Latvian State Research Programme NexIT Project No.~1.}}

\date{}

\begin{document}

\renewcommand\Affilfont{\small}

\author{Andris Ambainis}
\author{Kaspars Balodis}
\author{J\={a}nis Iraids}
\author{\authorcr Martins Kokainis}
\author{Kri\v{s}j\={a}nis Pr\={u}sis}
\author{Jevg\={e}nijs Vihrovs}
\affil{Centre for Quantum Computer Science, Faculty of Computing, \authorcr University of Latvia, Rai\c{n}a 19, Riga, Latvia, LV-1586.}

\maketitle

\begin{abstract}
In this paper we study quantum algorithms for NP-complete problems whose best classical algorithm is an exponential time application of dynamic programming. We introduce the path in the hypercube problem that models many of these dynamic programming algorithms. In this problem we are asked whether there is a path from $0^n$ to $1^n$ in a given subgraph of the Boolean hypercube, where the edges are all directed from smaller to larger Hamming weight.  We give a quantum algorithm that solves path in the hypercube in time $O^*(1.817^n)$. The technique combines Grover's search with computing a partial dynamic programming table. We use this approach to solve a variety of vertex ordering problems on graphs in the same time $O^*(1.817^n)$, and graph bandwidth in time $O^*(2.946^n)$. Then we use similar ideas to solve the travelling salesman problem and minimum set cover in time $O^*(1.728^n)$.
\end{abstract}

\section{Introduction}

Grover's quantum search \cite{Grover96} achieves a quadratic speedup over classical exhaustive search, but often naive exhaustive search is not the fastest classical algorithm. For example, for the Travelling Salesman Problem (TSP), exhaustive search over all the salesman’s routes runs in time $O^*(n!)$\footnote{The $O^*(f(n))$ notation hides a polynomial factor in $n$.} and Grover’s quantum search speeds this up to $O^*(\sqrt{n!})$. However, this is slower than the best classical algorithm which runs in time $O^*(2^n)$ and is based on dynamic programming \cite{Bel62,HK62}. This situation is quite typical: there are many classical algorithms for NP-complete problems that run in exponential time \cite{FK10} but are faster than a simple application of quantum search.

Can we combine quantum effects with classical problem-solving strategies for NP-complete problems? In this paper, we look at this question for dynamic programming algorithms. We find this interesting for at least two reasons. First, given the many applications of dynamic programming, this can lead to quantum speedups for many computational tasks. Second, there is something inherently sequential in the way dynamic programming solves problems, from smaller to larger subproblems. Such algorithms often store the solutions to all subproblems, which is a serious obstacle in quantizing them. Finding a quantum algorithm that improves over this is an interesting technical challenge.

The focus of this paper is on dynamic programming algorithms which solve subproblems that correspond to 
subsets of an $n$-element set, in order from smaller to larger subsets. Since there are $2^n$ subsets of 
an $n$-element set, this typically results in $O^*(2^n)$ time classical algorithms. 

A well-known example of an algorithm in this class is the famous Bellman-Held-Karp $O^*(2^n)$ time algorithm for the
Travelling Salesman Problem \cite{Bel62,HK62} (which is still the best known algorithm for this problem,
despite being more than 50 years old). It builds up optimal paths through subsets $S$ of the set of vertices, 
using knowledge of optimal paths for the subsets that contain one vertex less.
Note that if the edge weights are polylogarithmic in $n$, there is a more efficient $O^*(1.657^n)$ time algorithm \cite{Bjo14}.

\subsection{Our Results}

Inspired by the Bellman-Held-Karp algorithm, we first define a generic problem which captures this type of strategies for
dynamic problems, which we call Path in the Hypercube.
In this problem, we are given a subgraph $G$ of the Boolean hypercube, in which the edges are directed from smaller to larger Hamming weight.
The task is to determine whether $0^n$ and $1^n$ are connected by a path in $G$.

This problem can be solved classically in time $O^*(2^n)$ using dynamic programming, since a vertex $y$ is reachable from $0^n$ iff there is a path from $0^n$ to some $x$ such that there is an edge from $x$ to $y$.
Hence a quantum algorithm for this problem can be used to speed up classical Bellman-Held-Karp type algorithms that run in time $O^*(2^n)$.

We give a quantum algorithm for the Path in the Hypercube with running time $O^*(1.817^n)$.
We then use this algorithm to solve a variety of dynamic programming problems, including many graph vertex ordering problems, in the same time $O^*(1.817^n)$, while the best classical running time for these problems is $O^*(2^n)$.

Then we consider a variation of the Path in the Hypercube in which at most $\mu^n$ vertices of the hypercube are marked as available for the path, where $\mu$ is some constant.
We obtain a quantum speedup for this problem when $\mu \geq 1.735$ and then use it to solve Graph Bandwidth in time $O^*(2.945^n)$.
The best known classical algorithm for this problem has running time $O^*(4.383^n)$.

Next, we consider the subclass of problems where the solution to a problem that involves a set $S$ of size $n$ can be calculated by considering all partitions of $S$ into two sets of size $k$ and $n-k$ and choosing the best option, for any fixed positive $k$.
For example, TSP admits a decomposition of this type, as subpaths of an optimal path should also be optimal.
We show a quantum algorithm for TSP with running time $O^*(1.728^n \log L)$, where $L$ is the maximum edge weight.
Then we adapt this algorithm to solve Minimum Set Cover (MSC) and Feedback Arc Set in the same time $O^*(1.728^n)$.

As mentioned earlier, the fastest known classical algorithms for TSP on general graphs are the Bellman-Held-Karp algorithm that runs in time $O(n^2 2^n \log L)$ and Björklund's algorithm that runs in time $O^*(1.657^n L)$ \cite{Bjo14}.
For MSC, there are classical algorithms with running times $O(nm2^n)$, $O(n2^m)$ and $O(1.227^{n+m})$ \cite{RB08}.
Note that a trivial application of Grover's search results in an $O^*(\sqrt{2^m})$ quantum algorithm for MSC.
The fastest known classical algorithm for Feedback Arc Set runs in time $O^*(2^n)$.

Table \ref{tbl:res} lists the running times of the algorithms given in this paper.
All the algorithms require exponential space as they need to store a partial dynamic programming table.
The algorithms return the correct answer with probability at least $2/3$.

\begin{table}[H]
\begin{center}
\begin{tabular}{c | c | c}
& Classical (best known) & Quantum (this paper)\\ \hline
Path in the Hypercube & $O(n 2^n)$ & $O^*(1.817^n)$\\
Vertex Ordering Problems & $O^*(2^n)$ \cite{Bodlaender2012} & $O^*(1.817^n)$\\
Graph Bandwidth & $O^*(4.383^n)$ \cite{CP10} & $O^*(2.946^n)$ \\
Travelling Salesman Problem & $O(n^2 2^n)$ \cite{Bel62, HK62} & $O^*(1.728^n)$\\
Feedback Arc Set & $O^*(2^n)$ \cite{Bodlaender2012} & $O^*(1.728^n)$ \\
Minimum Set Cover & $O(nm 2^n)$ \cite{FK10} & $O(\poly(m,n)1.728^n)$
\end{tabular}
\end{center}
\caption{Summary of the results.} \label{tbl:res}
\end{table}

\subsection{Techniques}

The main idea of our algorithms is to precompute solutions for a part of the subsets using dynamic programming, and then use Grover's search on the rest of the subsets to find the answer to the problem.

As a warmup, consider the following simple quantum algorithm that checks whether a graph $G = (V,E)$ has a Hamiltonian cycle.
Let $f(S,u,v)$ be true iff there is a simple path that goes through all vertices of $S$, starts at $u$ and ends at $v$.
The algorithm first precomputes $f(S,u,v)$ for all $S$ such that $|S| = \frac{n}{4}+1$ using dynamic programming.
Then it runs Grover's search over all $S$ such that $|S| = \frac{n}{2}+1$ and $u, v \in S$ and checks whether $f(S,u,v)$ and $f((V \setminus S) \cup \{u, v\}, v, u)$ are both true (in that case a Hamiltonian cycle exists).
To find $f(S,u,v)$ for $|S| =\frac{n}{2}+1$, the algorithm runs another Grover's search over all subsets $T$ of $S$ such that $|T| = \frac{n}{4}+1, u \in T, v \notin T$ and vertices $t \in T$ to check whether $f(T,u,t)$ and $f((S\setminus T) \cup t,t,v)$ are both true (which we know from the classical preprocessing).
The complexity of this algorithm is $$O^*\left({n \choose \frac{n}{4}+1}+\sqrt{{n \choose \frac{n}{2}+1} {\frac{n}{2}+1 \choose \frac{n}{4}+1}}\right) = O^*\left(1.755^n\right).$$
This algorithm can be immediately extended for TSP by replacing Grover's search with quantum minimum finding \cite{DH96}.
A slight optimization of this algorithm gives complexity $O^*(1.728^n)$ which we describe in Section \ref{sec:tsp}.

\subsection{Prior Work}

Achieving a quantum advantage over classical dynamic programming algorithms has been a known problem in the quantum
algorithms community, with no substantial results on it up to now.

The first quantum improvement for TSP with an upper bound on its running time was given by \cite{MGA16}, who showed that for graphs with maximum degree $k$, there is a quantum algorithm that runs in time $O(2^{(k-2)n/4})$.
Quantum speedups for the Travelling Salesman Problem in graphs with small maximum degree were further studied by \cite{MLM95}.
They showed algorithms that solve TSP on degree-3 graphs in time $O^*(1.110^n)$, degree-4 graphs in $O^*(1.301^n)$, degree-5 and degree-6 graphs in $O^*(1.680^n)$ time.
In this case, the best classical algorithm is based on a different method, {\em backtracking} which performs a depth-first search on
a search tree of an unknown structure. The quantum algorithm follows by applying a result of Montanaro \cite{Montanaro15} 
who constructed a quantum backtracking algorithm with a nearly quadratic advantage over its classical counterparts (with further developments in \cite{AK17}). 
This work has no implications for TSP on general graphs.

\section{Preliminaries}

\paragraph{Notation.} If $f(n) = O(n^c)$ for some constant $c$, we will write $f(n) = \poly(n)$.
If $f$ polynomially depends on two parameters $m, n$, we will use $f(n) = \poly(m,n)$.
In case $f(n) = \poly(n)\cdot c^n$ for some constant $c$, we use notation $f(n) = O^*(c^n)$.
If $f(n) = O((\log n)^c f(n))$ for some constant $c$, we will write $\widetilde O(f(n))$.

\paragraph{Problems.}

Here we formally define the four problems we mainly focus on in this paper.

\begin{itemize}
\item \textbf{Path in the Hypercube.} 
Let the directed hypercube graph $Q_n$ be the graph formed by vertices numbered by $\{0,1\}^n$ and edges $x \to y$ for all $x, y$ such that $y$ is obtained from $x$ by changing one bit $x_i$ from 0 to 1.
In \textsc{Hypercube Path} we are given a query access to a subgraph $G$ of $Q_n$ (with queries that answer whether an edge $x\to y$ of $Q_n$ is present in $G$), and asked whether there is a path from $0^n$ to $1^n$.

\item \textbf{Graph Bandwidth.}
In the \textsc{Bandwidth} problem we are asked to find a linear ordering of the graph vertices such that the length of the maximum edge is minimized.
Formally, for a given graph $G = (V,E)$ on $n$ vertices, the task is to find
$$\min_{\pi \in S_n} \left\{ \max_{\{u, v\} \in E} \{ |\pi_u - \pi_v| \} \right\},$$
where $S_n$ is the set of all permutations on $n$ numbers.

\item \textbf{Travelling Salesman Problem.}
In \textsc{Travelling Salesman} we are given a weighted graph $G$ with $n$ nodes, and the task is to find the length of the shortest simple cycle that visits each vertex.

\item \textbf{Minimum Set Cover.}
In \textsc{Set Cover} the input is a collection $\mathcal S$ of subsets of an element universe $\mathcal U$ (denote $|\mathcal U| = n$ and $|\mathcal S| = m$), and the task is to find the minimum cardinality of a subset $\mathcal S' \subseteq \mathcal S$ such that $$\bigcup_{S \in \mathcal S'} S = \mathcal U.$$
\end{itemize}

\paragraph{Model.} Our algorithms work in the commonly used
QRAM (quantum random access memory) model of computation \cite{GLM08}
which assumes quantum memory which can be accessed in a superposition. 
QRAM has the property that any time-$T$ classical algorithm that uses random access memory
can be invoked as a subroutine for a quantum algorithm in time $O(T)$.
We can thus use primitives for quantum search (e.g., Grover's quantum search or minimum finding) 
with conditions checking which requires data stored in a random access memory.

\paragraph{Tools.} We will use the following well-known results in our algorithms.

\begin{theorem}
\label{thm:vts}
[Variable Time Search \cite{Amb10}]
Let $\mathcal A_1, \ldots, \mathcal A_n$ be quantum algorithms that return true or false and run in unknown times $T_1, \ldots, T_n$, respectively.
Suppose that each $\mathcal A_i$ outputs the correct answer with probability at least $2/3$.
Then there exists a quantum algorithm with success probability at least $2/3$ that checks whether at least one of $\mathcal A_i$ returns true and runs in time
$$\widetilde O\left(\sqrt{T_1^2+\ldots+T_n^2}\right).$$
\end{theorem}

If we search among $n$ elements, each of which can be checked quantumly in known time $T$, we need time $O(T \sqrt n)$ (without polylogarithmic factors).
In that case we will simply say that we perform Grover's search over these $n$ elements \cite{Grover96}.

\begin{theorem}[Quantum Minimum Finding \cite{DH96}]
Let $a_1, \ldots, a_n$ be integers, accessed by a procedure $\mathcal P$.
There exists a quantum algorithm that finds $\min_{i=1}^n \{a_i\}$ with success probability at least $2/3$ using $O(\sqrt n)$ applications of $\mathcal P$.
\end{theorem}

Suppose that $a_i$ is given by the output of a subroutine $\mathcal A_i$, which is correct with probability at least $2/3$.
In that case we will implement $\mathcal P$ as $O(\log n)$ repetitions of $\mathcal A_i$, to reduce the error probability of each $\mathcal P$ call to $O(1/\sqrt n)$.
Then with probability at least $2/3$, all the calls to $\mathcal P$ will be correct, and quantum minimum finding will also be correct with probability at least $2/3$, resulting in a total $O(\sqrt n \log n)$ runtime.

For the running time analysis of exponential time algorithms, we use the following approximation of binomial coefficients:
\begin{theorem}[Entropy Approximation] \label{entropy}
For all $1 \leq k \leq n/2$,
$${n \choose \leq k} = \sum_{i=0}^k {n \choose k} \leq 2^{\be(k/n) \cdot n},$$
where $\be(\epsilon) = -(\epsilon \log_2(\epsilon) +  (1-\epsilon) \log_2(1-\epsilon))$ is the binary entropy of $\epsilon \in [0,1]$.
\end{theorem}

For $k>n/2$, we use the trivial approximation ${n \choose \leq k}\leq 2^n$.
For single binomial coefficients, we use ${n \choose k} \leq 2^{\be(k/n) \cdot n}$ (which applies to all $1 \leq k \leq n$).

\section{Path in the Hypercube}

In this section, we describe the quantum algorithm for finding a path in a subgraph $G$ of the hypercube
(which models many vertex ordering problems, as described in Section \ref{sec:order}).
Let $\mathcal P_{\text{edge}}$ be the procedure that checks whether there is an edge between two vertices $x, y \in \{0, 1\}^n$ in $G$, and assume that one query to $\mathcal P_{\text{edge}}$ requires $\poly(n)$ time.

\begin{theorem}
There is a bounded-error quantum algorithm that solves \textsc{Hypercube Path} in time $O^*(1.817^n)$.
\end{theorem}

Let $k \geq 2$ be a parameter of the algorithm and $0 < \alpha_1 < \alpha_2 < \ldots < \alpha_k < \alpha_{k+1} = 1/2$ be constants to be defined later.
Let $A_i$ be the set of vertices with Hamming weight $\lfloor\alpha_i n\rfloor$.
For two Boolean strings $x$ and $y$, we write $x \preceq y$ iff $x_i \leq y_i$ for all $i \in [n]$.

\begin{algorithm}[H]
\textbf{HypercubePath}(subgraph $G$ of $Q_n$): whether $0^n$ and $1^n$ are connected by a directed path.

\begin{enumerate}
\item \label{itm:hc1} If $\lfloor\alpha_in\rfloor = \lfloor\alpha_{i+1}n\rfloor$ for some $i$, run the classical dynamic programming algorithm and return the answer.
\item Otherwise, perform a classical preprocessing, where we classically find whether each vertex $x \in A_1$ is reachable from $0^n$ using dynamic programming. The answer is given by a recursive formula $$r(x) = \bigvee_{i : x_i=1} \left(r(x^i) \land (x^i,x) \in G\right), \hspace{1cm} r(0^n) = \text{true},$$ where $x^i$ is $x$ with its $i$-th bit flipped.
Store the values $r(x)$ for all $x \in A_1$ in memory. \\
We perform a similar preprocessing to also compute from which vertices with Hamming weight $n-\lfloor\alpha_1n\rfloor$ we can reach $1^n$.
\item
Let $\textbf{Path}_b(x), b\in\{0, 1\}$ be a quantum subroutine (defined below) that returns true iff there is a path between $b^n$ and $x$. We run Grover's search over all vertices $x$ in the middle level of the hypercube $A_{k+1}$, searching for an $x$ 
for which both $\textbf{Path}_0(x)$ and $\textbf{Path}_1(x)$ return true.
\end{enumerate}

Now we describe how to implement $\textbf{Path}_0(x)$ ($\textbf{Path}_1(x)$ is implemented similarly). 
Let $\textbf{Reachable}_i(x)$ be a quantum procedure that computes whether $x \in A_i$ is reachable from $0^n$:

\begin{enumerate}
\item If $i = 1$, then $\textbf{Reachable}_i(x) = r(x)$.
\item If $i > 1$, we use Grover's search to search for a vertex $y \in A_{i-1}$, $y \preceq x$ with directed paths $0^n \rightarrow y$ and $y\rightarrow x$:
\begin{itemize}
\item Call $\textbf{Reachable}_{i-1}(y)$ to check if there is a path from $0^n$ to $y$; 
\item
Call \textbf{HypercubePath} recursively on the subgraph of a smaller subcube $\{z \in \{0, 1\}^n \mid y \preceq z \preceq x\}$
to check if there is a path from $y$ to $x$.
\end{itemize}
\end{enumerate}
Then $\textbf{Path}_0(x) = \textbf{Reachable}_{k+1}(x)$.
\caption{Quantum algorithm for \textsc{Hypercube Path}}
\end{algorithm}

\subsection{Running Time} \label{sec:hctime}

The recursive depth of the presented algorithm is $d=O(\log n)$, since the dimension of the hypercube decreases $1/(\alpha_i-\alpha_{i-1})$ times with each recursive call.

We express the running time as $O(n^{f(n)} \gamma^n)$ for some constant $\gamma$ and $f(n) = o(n/\log n)$ and determine the dominant factor of the running time, $\gamma^n$ (which we call the {\em exponential part} of the running time).
This results in omitting a multiplicative factor of $O(n^c)$
at each level of recursion. If the complexity without these factors is $O(\gamma^n)$, the actual complexity
is $O(\gamma^n n^{O(\log n)})=O^* ((\gamma+\epsilon)^n)$ for any $\epsilon > 0$.

We now find the value of $\gamma$.
\begin{itemize}
\item Let $T_i$ be the exponential part of the running time of $\textbf{Reachable}_i(x)$.
It is given by the recurrence with $T_1 = 1$ and
\begin{equation} \label{eq:reachable}
T_{i+1} = \sqrt{{\alpha_{i+1}n \choose \alpha_in}}\left(\gamma^{(\alpha_{i+1}-\alpha_i)n} + T_i\right).
\end{equation}
\item The classical preprocessing step (computing $r(x)$ for all $x$ with Hamming weight at most $\alpha_1 n$), requires time ${n \choose \leq \alpha_1 n}$.
\item The exponential part of the running time of the whole algorithm is 
$${n \choose \leq \alpha_1 n}+ \sqrt{{n \choose n/2}} T_{k+1}.$$
\end{itemize}

To find $\gamma$, we need to balance the summands.
For this, we require the following constraints:
\begin{itemize}
\item $\gamma^n = {n \choose \leq \alpha_1 n}$;
\item $\gamma^n = \sqrt{{n \choose n/2}} T_{k+1}$;
\item $\gamma^{(\alpha_{i+1}-\alpha_i)n} = T_i$ for all $i \in [2;k]$.
\end{itemize}

By using Eq.~(\ref{eq:reachable}), we obtain the following system of equations:
\begin{equation} \label{eq:req}
\begin{cases}
\gamma^n = {n \choose \leq \alpha_1 n} \\
\gamma^n = \sqrt{{n \choose n/2}{n/2 \choose \alpha_k n}} \gamma^{(1/2-\alpha_k)n} \\
\gamma^{(\alpha_{i+1}-\alpha_i)n} = \sqrt{{\alpha_in \choose \alpha_{i-1}n}} \gamma^{(\alpha_i-\alpha_{i-1})n} \hspace{1cm} \text{for all $i \in [2;k]$}.
\end{cases}
\end{equation}
By letting $\gamma = 2^c$ and applying Theorem \ref{entropy}, this reduces to:
$$
\begin{cases}
c = \be(\alpha_1) \\
2c(2\alpha_k + 1) = 2+\be(2\alpha_k) \\
2c(\alpha_{i+1}-2\alpha_i+\alpha_{i-1}) = \alpha_i \be\left(\frac{\alpha_{i-1}}{\alpha_i}\right)  \hspace{1cm} \text{for all $i \in [2;k]$}.
\end{cases}
$$
By solving these equations numerically for fixed $k \geq 2$, we get increasingly better results for $\gamma$.
For $k = 6$, we obtain $\gamma \approx 2^{0.861483} \approx 1.816905$ and further increments of $k$ lead to negligible improvements.
The solution for $\alpha = (\alpha_1, \ldots, \alpha_6)$ then is $$\alpha \approx (0.28448, 0.28453, 0.28470, 0.28628, 0.29604, 0.34174).$$

\subsection{Vertex Ordering Problems}
\label{sec:order}

The presented algorithm can be used for many graph vertex ordering problems.
Consider any problem in which, for a given graph $G=(V,E)$, we have to compute a value in one of the following forms:
\begin{equation}
\label{eq:ord}
\min_{\pi \in S_n} \max_{v \in V} f(G, \pi_{<v},v) \hspace{1cm} \text{or} \hspace{1cm} \min_{\pi \in S_n} \sum_{v \in V} f(G, \pi_{<v},v),
\end{equation}
where $S_n$ is the set of all $n$-permutations, $\pi_{<v} = \{u \in V \mid \pi_u < \pi_v\}$, and $f(G, S, v)$ is computable in $\poly(n)$ time, given the graph $G$, a subset $S\subset V$ and a vertex $v\notin S$.

\cite{Bodlaender2012} prove that such problems can be computed by dynamic programming on sets in $O^*(2^n)$ time using the Bellman-Held-Karp algorithm.
The corresponding recursive formulas for a set $S \subseteq V$ are given by
$$A_G(S) = \min_{v \in S} \max\{f(G,S\setminus \{v\},v), A_G(S \setminus \{v\})\}$$
and
$$A_G(S) = \min_{v \in S} (f(G,S\setminus \{v\},v) + A_G(S \setminus \{v\})).$$
By replacing Grover's search with quantum minimum finding in our \textbf{HypercubePath} algorithm, we obtain a quantum speedup for these problems.

\begin{corollary}
For any problem that can be expressed in any of the forms given in Eq.~(\ref{eq:ord}), there is a bounded-error quantum algorithm that solves it in time $O^*(1.817^n)$.
\end{corollary}

Examples include \textsc{Treewidth}, \textsc{Minimum Fill-In}, \textsc{Pathwidth}, \textsc{Sum Cut}, \textsc{Minimum Interval Graph Completion}, \textsc{Cutwidth} and \textsc{Optimal Linear Arrangement} (see \cite{Bodlaender2012} for the descriptions of these problems).
Note that for \textsc{Treewidth} and \textsc{Minimum Fill-In} there are more efficient classical algorithms that run in $O^*(1.735^n)$ time, see \cite{FTV15}.

Similar ideas also give an $O^*(1.817^n)$ quantum algorithm  for \textsc{Travelling Salesman} but we present a better quantum algorithm for this problem in Section \ref{sec:tsp}, with running time of $O^*(1.728^n)$.

\subsection{Path Finding with Forbidden Vertices} \label{sec:forbidden}

In some problems, the entries of the dynamic programming table are indexed by subsets of $[n]$ but we only need to compute entries corresponding to some subsets. An example is the algorithm for \textsc{Bandwidth} in Section \ref{sec:band}. 

That corresponds to a setting where some vertices in the given hypercube are invalid, i.e., the path is not allowed to go through such vertices. We assume that the validity of a vertex $x$ 
can be checked by the procedure $\mathcal P_{\text{valid}}$, which we can usually assume to run in $\poly(n)$ time.
Let $\mu^n$ be  the total number of valid vertices, with $\mu$ not known to the algorithm in the beginning.

Let $\mathcal A_x$ be a procedure that for a vertex $x \in A_{k+1}$ runs $\textbf{Path}_0(x)$ and $\textbf{Path}_1(x)$ only after checking whether $x$ is valid using $\mathcal P_{\text{valid}}$.
Instead of Grover's search over the middle layer of the hypercube in \textbf{HypercubePath} algorithm, we now run variable time search over $\mathcal A_x$.
A single call to any $\mathcal A_x$ takes either $\poly(n)$ time (if $x$ is invalid) or time $T_{k+1}$.
We note that recursive calls to \textbf{HypercubePath} of dimension $n'$ still require time $\gamma^{n'}=1.817^{n'}$, as $\mu^n$ is an upper bound for the valid vertices in the whole hypercube, and if $2^{n'} < \mu^n$, there is no guaranteed improvement in the recursive call.

Then the exponential part of the running time of the quantum part (excluding the classical preprocessing) is at most
\begin{equation}
\label{eq:runtime} \sqrt{2^n + \mu^n \cdot T_{k+1}^2} = O(\sqrt{\mu^n} T_{k+1}).
\end{equation}

Let $\gamma_\mu^n$ be the exponential part of the complexity of the whole algorithm (e.g., $\gamma_2 \approx 1.817$).
We can find the optimal $\gamma_\mu$ by solving the following system of equations:
$$
\begin{cases}
\gamma_\mu^n = {n \choose \leq \alpha_1 n} \\
\gamma_\mu^n = \sqrt{\mu^n {n/2 \choose \alpha_k n}} \gamma_2^{(1/2-\alpha_k)n} \\
\gamma_2^{(\alpha_{i+1}-\alpha_i)n} = \sqrt{{\alpha_in \choose \alpha_{i-1}n}} \gamma_2^{(\alpha_i-\alpha_{i-1})n} \hspace{.8cm} \text{for all $i \in [2;k]$}.
\end{cases}
$$
For example, if $\mu = 1.8$, we get $\gamma_{1.8} \approx 1.7568$.

Note that the classical depth-first search algorithm that traverses all valid vertices, starting from $0^n$, solves the problem in time $O^*(\mu^n)$.
Figure \ref{fig:graph} shows the performance of the quantum algorithm compared to the classical.
They have the same complexity at the point $\mu_0 \approx 1.734622$.

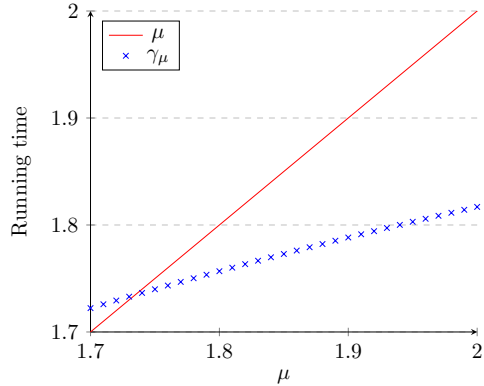
\begin{figure}[H]
\begin{center}
\begin{tikzpicture}[scale=0.75]
\begin{axis}[
    axis lines = left,
    xlabel={$\mu$},
    ylabel={Running time},
    xmin=1.7, xmax=2,
    ymin=1.7, ymax=2,
    xtick={1.7,1.8,1.9,2},
    ytick={1.7,1.8,1.9,2},
    ymajorgrids=true,
    grid style=dashed,
    legend pos=north west
]
\addplot[color=red]{x};
    \addlegendentry{$\mu$}

\addplot[
    only marks,
    color=blue,
    mark=x,
    ]
    coordinates {
    (1.7,1.7223)
    (1.71,1.7259)
    (1.72,1.7294)
    (1.73,1.7330)
    (1.74,1.7365)
    (1.75,1.7399)
    (1.76,1.7434)
    (1.77,1.7468)
    (1.78,1.7502)
    (1.79,1.7535)
    (1.8,1.7568)
    (1.81,1.7601)
    (1.82,1.7634)
    (1.83,1.7666)
    (1.84,1.7698)
    (1.85,1.7729)
    (1.86,1.7761)
    (1.87,1.7792)
    (1.88,1.7822)
    (1.89,1.7853)
    (1.9,1.7883)
    (1.91,1.7913)
    (1.92,1.7942)
    (1.93,1.7972)
    (1.94,1.8001)
    (1.95,1.8030)
    (1.96,1.8058)
    (1.97,1.8086)
    (1.98,1.8114)
    (1.99,1.8142)
    (2,1.8169)
    };
    \addlegendentry{$\gamma_\mu$}
\end{axis}
\end{tikzpicture}
\end{center}
\caption{Running time of classical and quantum algorithms for \textsc{Hypercube Path} with $\mu^n$ valid vertices.} \label{fig:graph}
\end{figure}

\section{Graph Bandwidth}
\label{sec:band}

In this section, we show how it is possible to speed up Cygan's and Pilipczuk's $O^*(5^n)$ algorithm from \cite{CP08} using the \textbf{HypercubePath} algorithm with forbidden vertices.

\begin{theorem}
There is a bounded-error quantum algorithm that solves \textsc{Bandwidth} in time \linebreak $O^*(2.945^n)$.
\end{theorem}

First we quickly recap the classical algorithm itself.

\subsection{Classical Algorithm}

We assume that the given graph is connected, as each connected component can be processed independently.
The algorithm performs binary search on the answer to find the optimal bandwidth.
For a fixed $b \in [n-1]$, the main procedure checks whether the bandwidth is at most $b$.

For simplicity, assume that $(b+1) \mid n$.
The vertices are placed into $k = n/(b+1)$ buckets of $b+1$ vertices in each.
The $k$ buckets are numbered 1 to $k$ and placed in order. The $i$-th bucket contains vertex positions $[(i-1)(b+1)+1, i(b+1)]$. For the bandwidth of any ordering to not exceed $b$, every pair of connected vertices must be in the same or adjacent buckets.
The algorithm considers all possible \emph{assigments} of vertices to buckets.

The assignments are constructed as follows.
First pick any spanning tree $T$ of $G$ and root this tree at vertex $r$.
For a particular assignment $\sigma : [n] \to [k]$ assign $r$ to any bucket $\sigma(r)$.
Then, for each vertex $v \neq r$, choose a number $s_{\sigma}(v) \in \{-1, 0, +1\}$.
This denotes the offset between the bucket of $v$ and its parent $p(v)$: $\sigma(v) = \sigma(p(v)) + s_{\sigma}(v)$.
Some of these assignments do not place the correct number of vertices in the buckets or assign vertices to buckets outside $[k]$, but every correct assignment can be represented this way.

For a fixed correct bucket assignment $\sigma$, the algorithm first checks whether there are any edges between two vertices $u, v$ such that $|\sigma(u) - \sigma(v)| \geq 2$. If so, the bandwidth of this assignment already exceeds $b$. If not, any edges of length greater than $b$ can only be between adjacent buckets. The algorithm then checks whether it has a feasible \emph{vertex ordering} using dynamic programming.

\paragraph{Dynamic Programming.} The algorithm examines subsets $S \subseteq [n]$ with the following property: if $S_i = \{v \in S \mid \sigma(v) = i\}$ are the vertices from $S$ assigned to the $i$-th bucket, then:
$|S_1|\geq |S_2| \geq \ldots \geq |S_k|$ and $|S_k| \geq |S_1|-1$.

The dynamic programming computes, for each such set $S$, whether there is an ordering of $S$ such that
\begin{itemize}
\item the vertices $S_i$ can be assigned to the first $|S_i|$ positions of the bucket: $[(i-1)(b+1)+1,(i-1)(b+1)+|S_i|]$;
\item the bandwidth of this partial ordering is at most $b$ (for each $v \in S$, any edge $\{u,v\}$ has length $\leq b$).
\end{itemize}

Now let $i$ be the last bucket of maximum size in $S$; then for each $v \in S_i$, we check whether $v$ can be placed at position $(i-1)(b+1)+|S_i|$ to form a partial ordering with bandwidth at most $b$.
This is possible if and only if the following conditions hold:
\begin{enumerate}
\item \label{itm:c2} there are no edges between $v$ and $S_{i-1}$; 
\item \label{itm:c3} there are no edges between $v$ and $\sigma^{-1}(i+1) \setminus S_{i+1}$,
\item \label{itm:c1}$S \setminus \{ v \}$ admits a partial ordering with bandwidth at most $b$.
\end{enumerate}
The first two conditions describe edges adjacent to $v$ that can be longer than $b$. These two conditions can be checked in $\poly(n)$ time.
The third condition can be checked using the result for $S\setminus \{ v \}$ from a previous step of the computation.

\paragraph{Running Time.}

For a fixed $\sigma$, a set $S \subseteq [n]$ admitting a partial ordering with bandwidth at most $b$ must satisfy the following conditions for every $v \neq r$:
\begin{itemize}
\item if $p(v) \in S$, then it cannot be that $s_{\sigma}(v) = +1$ and $v \notin S$; 
\item if $p(v) \notin S$, then it cannot be that $s_{\sigma}(v) = -1$ and $v \in S$.
\end{itemize}
Call a pair $(\sigma, S)$ \emph{valid} if it satisfies these conditions. We can upper bound the running time of the algorithm by bounding the number of valid pairs.

Once we fix whether or not $p(v) \in S$, there are only 5 valid choices for $(s_{\sigma}(v), [v \in S])$. Additionally there are $k$ choices for $\sigma(r)$ and 2 choices for whether $r \in S$, hence the algorithm has to examine at most $2k \cdot 5^n$ valid pairs $(\sigma, S)$ and the running time is $O^*(5^n)$.

\subsection{Quantum Algorithm}

For a fixed assignment $\sigma$, the problem of determining whether there is a compatible vertex ordering with bandwidth at most $b$ can be seen as \textsc{Hypercube Path} with forbidden vertices.
For any set $S$, the procedure $\mathcal P_{\text{valid}}$ checks whether the pair $(\sigma, S)$ is valid. The procedure $\mathcal P_{\text{edge}}$ checks whether for a subset $P$ of $S$ such that $|P|+1=|S|$, we have that the vertex $S \setminus P$ is placed in the last bucket of maximum size in $S$ and satisfies conditions \ref{itm:c2} and \ref{itm:c3}.

Let the number of valid sets for an assignment $\sigma$ be $\mu(\sigma)^n$.
The classical depth-first search over valid assignments decides whether $\sigma$ has a valid ordering in time $\mu(\sigma)^n$.
On the other hand, the quantum algorithm from Section \ref{sec:forbidden} works in exponential time $\gamma_{\mu(\sigma)}^n$.
We know that $\mu_0$ is the point where the quantum algorithm becomes faster than the classical one.
Then to solve the problem for a fixed $\sigma$, we can run the depth-first search until it terminates or visits $\mu_0^n$ valid vertices; if it has not terminated, we stop and run the quantum algorithm instead, which is then guaranteed to be faster.

The quantum algorithm for \textsc{Bandwidth} runs variable time search over all possible assignments $\sigma$, and for each $\sigma$ applies the subroutine described above.

\subsection{Running Time}

By Theorem \ref{thm:vts}, the exponential part of the running time of the algorithm is
$$\sqrt{\sum_{\sigma} \min\left\{\mu(\sigma)^{2n},\gamma_{\mu(\sigma)}^{2n}\right\}}.$$
Now we need to estimate the largest possible value of this expression. Let $A=\sum_{\sigma} \mu(\sigma)^n$ be the total number of valid pairs $(\sigma, S)$.
We claim that the maximum value of the expression is $O(\sqrt{A \cdot \mu_0^n})$.

We can reformulate this task as follows.
We have  real variables $x_1, \ldots, x_N$ 
such that $N \leq 3^n$, $0 \leq x_i \leq 2$ for each $i \in N$, and $\sum_{i=1}^N x_i^n = A$.
Let $I_c = \{i \in [N] \mid x_i \leq \mu_0\}$ and $I_q = \{i \in [N] \mid x_i > \mu_0\}$.
We need to maximize the following expression:
$$\sqrt{\sum_{i \in I_c} x_i^{2n} + \sum_{i \in I_q} \gamma_{x_i}^{2n}}.$$
For this, we upper bound each of the summands by $A \cdot \mu_0^n$.
\begin{itemize}
\item As $(x_i+x_j)^2 \geq x_i^2 + x_j^2$ for any $i, j$, the worst case for the first part is when as many $x_i$ are as large as possible, i.e., $x_i = \mu_0$.
Hence $$\sum_{i \in I_c} x_i^{2n} \leq \frac{A}{\mu_0^n} \cdot \mu_0^{2n} = A\cdot \mu_0^n.$$
\item For the second part, we need the following proposition:

\begin{proposition} \label{thm:gammamu}
For all $\mu \in [\mu_0, 2]$, we have $\gamma_{\mu}^2/\mu \leq \mu_0$.
\end{proposition}

\begin{proof}
Examine any level assignment $\alpha = (\alpha_1, \ldots, \alpha_k)$ for the \textbf{HypercubePath} quantum algorithm with $\mu_0^n$ valid vertices.
Let the running time of this algorithm be $\sqrt{\mu_0^n} T_{k+1}$, where $T_{k+1}$ is the complexity of the \textbf{Path} procedure with $k$ levels defined by $\alpha$, see Eq.~(\ref{eq:runtime}).
Let $\mu = q\mu_0$ and examine the \textbf{HypercubePath} algorithm with $\mu^n$ valid vertices.
Since we can use the assignment $\alpha$ for this algorithm as well, its running time is at most $\sqrt{\mu^n} T_{k+1} = \sqrt{q^n} \sqrt{\mu_0^n} T_{k+1}$.
Hence, $\gamma_{\mu} \leq \sqrt q\mu_0$.
Therefore, $$\gamma_{\mu}^2/\mu \leq (\sqrt q \mu_0)^2/(q\mu_0) = \mu_0. \qedhere$$
\end{proof}

Let $w(\mu)\cdot \frac{A}{\mu^n}$ be the number of $i \in [N]$ such that $x_i = \mu$, so that $\sum_{\mu} w(\mu) = 1$.
Then we have that
$$\sum_{i \in I_q} \gamma_{x_i}^{2n} = \sum_{\mu > \mu_0} w(\mu) \cdot \frac{A}{\mu^n} \cdot \gamma_{\mu}^{2n} \leq A \cdot \mu_0^n$$
by Proposition \ref{thm:gammamu}.
\end{itemize}

From the classical algorithm we have that $A = O^*(5^n)$.
Thus the algorithm runs in exponential time $\sqrt{5^n \cdot \mu_0^n} \approx 2.945^n$.
Once again, because of the additional polynomial factors, the final complexity of the quantum algorithm for \textsc{Bandwidth} is $O^*(2.945^n)$.

\section{Travelling Salesman Problem} \label{sec:tsp}

In this section we describe the quantum algorithm for \textsc{Travelling Salesman}.
It uses ideas from the Bellman-Held-Karp algorithm \cite{HK62} which uses dynamic programming and runs in time $O(n^22^n)$.

Let the given graph be $G = (V,E,w)$, where $|V| = n$, $E \subseteq V^2$ and $w : E \to \mathbb N$ are the edge weights (let $w(u,v) = \infty$ if $\{u,v\} \notin E$).
In this setting we assume that access to $w$ takes $\poly(n)$ time.
Let $L = \max_{\{u,v\} \in E} w(u,v)$.
Then addition of two integer numbers requires $O(\log L)$ time.

\begin{theorem}
There is a bounded-error quantum algorithm that solves \textsc{Travelling Salesman} in time $O^*(1.728^n \log L)$.
\end{theorem}

Let $D = \{(S, u, v) \mid S \subseteq V, u, v \in S\}$ and define $f : D \to \mathbb N$ as follows: $f(S, u, v)$ is the length of the shortest path in the graph induced by $S$ that starts in $u$, ends in $v$ and visits all vertices in $S$ exactly once.
Let $N(u)$ be the set of neighbours of $u$ in $G$.
We can calculate $f(S,u,v)$ by the following recurrence:
\begin{equation} \label{eq:rec1}
f(S,u,v) = \min_{t \in N(u) \cap S \atop t \neq v}\left\{ w(u, t) + f(S\setminus \{u\}, t, v) \right\}, \hspace{1cm}  f(\{v\},v,v) = 0.
\end{equation}
Note that $f(S,u,v)$ can also be calculated recursively by splitting $S$ into two sets.
Let $k \in [2,|S|-1]$ be some fixed number, then
\begin{equation} \label{eq:rec2}
f(S,u,v) = \min_{X \subset S, |X| = k \atop u \in X, v \notin X} \min_{t \in X \atop t \neq u} \{ f(X,u,t) + f((S\setminus X) \cup \{t\},t,v)\}.
\end{equation}

Let $\alpha \in (0,1/2]$ be a parameter of the algorithm to be specified later.

\begin{algorithm}[H]
\textbf{TravellingSalesman}(graph $G$, edge weights $w$): length of the shortest Hamiltonian cycle.
\begin{enumerate}
\item Calculate the values of $f(S,u,v)$ for all $|S| \leq (1-\alpha) n/4$ classically using dynamic programming with Eq.~(\ref{eq:rec1}) and store them in memory.
\item Run quantum minimum finding over all subsets $S \subset V$ such that $|S| = n/2$ to find the answer,
$$\min_{S \subset V \atop |S| = n/2} \min_{u, v \in S \atop u \neq v} \{ f(S,u,v) + f((V \setminus S) \cup \{u, v\}, v, u) \}.$$
To calculate $f(S,u,v)$ for $|S| = n/2$, run quantum minimum finding for Eq.~(\ref{eq:rec2}) with $k = n/4$.
To calculate $f(S,u,v)$ for $|S| = n/4$, run quantum minimum finding for Eq.~(\ref{eq:rec2}) with $k = \alpha n/4$.
For any $S$ such that $|S| = \alpha n/4$ or $|S| = (1-\alpha)n/4$, we know $f(S,u,v)$ from the classical preprocessing.
\end{enumerate}
\caption{Quantum algorithm for \textsc{Travelling Salesman}}
\end{algorithm}

\subsection{Running Time} \label{sec:tsptime}

The classical preprocessing requires time $$O^*\left({n \choose \leq (1-\alpha) n/4}\right) = O^*\left(2^{\be\left(\frac{1-\alpha}{4}\right)n}\right).$$
The complexity of the quantum part is $$O^*\left(\sqrt{{n \choose n/2} {n/2 \choose n/4} {n/4 \choose \alpha n/4}}\right) = O^*\left(2^{\frac 1 2 \left( 1 + \frac1 2 + \frac{\be(\alpha)}{4}\right)n}\right).$$
The overall complexity is minimized when both parts are equal.
The optimal choice for $\alpha$ then is approximately $0.055362$.
The running time of the algorithm then is $O^*(2^{0.788595n}) = O^*(1.727391...^n)$.
Arithmetic operations on integers give additional $O(\log L)$ multiplicative factor.

Splitting the vertex set $V$ into more than 2 sets results in a worse running time, and taking more than three levels of recursion does not seem to give any improvements.

\subsection{Feedback Arc Set}

Here we show another application of the presented algorithm.
It can be immediately extended to the problem \textsc{Directed Feedback Arc Set}.
In this problem we are given a directed graph $G = (V,E)$ and the task is to find the smallest subset of edges $F \subseteq E$ such that the graph $G' = (V,E\setminus F)$ is acyclic.
Let us call $F$ the minimum feedback arc set.

Let us formulate the recursive formula for this problem such that we can apply the Travelling Salesman Problem algorithm.
Consider any partition of $V$ into two sets $A$ and $B$ of size $k$ and $n-k$, respectively.
We require that in the acyclic graph $G'$ there are no edges going from $B$ to $A$, i.e., $A$ is before $B$ in the topological order of $G'$.
The minimum feedback sets for the sets $A$ and $B$ can be computed independently.
Let $f(S)$ be the size of the minimum feedback arc set of $G$ induced on $S$.
Then we can compute this value by the following recurrence:
$$f(S) = \min_{A \subseteq S \atop |A| = k} \{ f(A) + f(S \setminus A) + |\{(a,b) \in E \mid a \in A, b \in B\}|\}.$$

We can use this recurrence for the quantum algorithm in the same way the \textbf{TravellingSalesman} algorithm utilizes Eq.~(\ref{eq:rec2}).
Hence, we have the following result.

\begin{corollary}
There is a bounded-error quantum algorithm that solves \textsc{Directed Feedback Arc Set} in time $O^*(1.728^n)$.
\end{corollary}

The best known classical algorithm for this problem is an $O^*(2^n)$ time application of dynamic programming \cite{Bodlaender2012}.

\section{Set Cover}

In this section, we show a quantum speedup for \textsc{Set Cover} using an idea similar to the \textsc{Travelling Salesman} algorithm.

\begin{theorem}
There is a bounded-error quantum algorithm that solves \textsc{Set Cover} in time \linebreak $O(\poly(m,n) 1.728^n)$.
\end{theorem}

Let $\mathcal C \subseteq \mathcal S$ be a minimum set cover of $\mathcal U$.
We call a set $P$ its \emph{partial} cover if $P = \bigcup_{S \in \mathcal C'} S$ for some $\mathcal C' \subseteq \mathcal C$.
Let $d = \beta n$ be a parameter of the algorithm, $0 < \beta < 1$.
We also define \emph{bands} for $i \in \{1,2,3\}$,
$$B_i = \{X \subseteq \mathcal U \mid ||X|-w_in|\leq d/2\},$$
where $w_3 = 1/2$, $w_2 = 1/4$ and $w_1 = \alpha/4$ for some parameter $\alpha \in (0,1/2]$ to be specified later.

The main idea of the algorithm is that if $|S| \leq d$ for all sets $S \in \mathcal S$, then each of the bands contains a partial cover of any minimum set cover.
Then for each band, we can search for the best partial cover using Grover's search.
This is done similarly as in the \textsc{Travelling Salesman} algorithm.

Let $f(\mathcal U, \mathcal S)$ to be the size of the minimum set cover of $\mathcal U$ using sets of $\mathcal S$.
The algorithm relies on the following fact.

\begin{proposition}
If $P$ is a partial cover of a minimum set cover of $\mathcal U$ using sets of $\mathcal S$, then
$$f(\mathcal U,\mathcal S) = f(P,\mathcal S) + f(\mathcal U \setminus P,\mathcal S).$$
\end{proposition}

\begin{proof}
Suppose that $P$ is a partial cover of a minimum set cover $\mathcal C$.
Let $\mathcal C(P) = \{S \in \mathcal C \mid S \subseteq P\}$ and $\mathcal C(\mathcal U \setminus P) = \mathcal C \setminus \mathcal C(P)$.
Also, let $\mathcal C'(P)$ and $\mathcal C'(\mathcal U \setminus P)$ be the minimum covers of $P$ and $\mathcal U \setminus P$ using $\mathcal S$, respectively.
Since $|\mathcal C'(P)| \leq |\mathcal C(P)|$ and $|\mathcal C'(\mathcal U \setminus P)| \leq |\mathcal C(\mathcal U \setminus P)|$, we have $f(P,\mathcal S) + f(\mathcal U \setminus P,\mathcal S) \leq |\mathcal C|$.
Moreover, $\mathcal C' = \mathcal C'(P) \cup \mathcal C'(\mathcal U \setminus P)$ is a set cover of $\mathcal U$, therefore it is also a minimum set cover.
\end{proof}

If there are sets $S$ such that $|S| > d$, we recursively solve the minimum set cover of $\mathcal U \setminus S$ using sets $\mathcal S \setminus \{S\}$.
Let $N$ be the number of elements of $\mathcal U$ in the very first call of the algorithm.
Suppose that the quantum algorithm above runs in time $O^*(\delta^n)$ for some $\delta = 2^c$.
Then if at any point $|\mathcal U|$ becomes less than $c N$, we run the classical dynamic programming algorithm instead (which runs in time $O^*(2^n)$).

\begin{algorithm}[H]
\textbf{MinCover}(universe $\mathcal U$, collection of sets $\mathcal S \subseteq 2^{\mathcal U}$): the size $|\mathcal C|$ of a minimum set cover $\mathcal C \subseteq \mathcal S$ of $\mathcal U$.
\begin{enumerate}
\item \label{itm:sc1} If $n = |\mathcal U| < c N$, find the minimum set cover using classical dynamic programming and return the answer.
\item \label{itm:sc2} Iterate over all sets $S \in \mathcal S$.
For all sets $S$ such that $|S| \geq d = \beta n$, find $\textbf{MinCover}(\mathcal U \setminus S, \{T \setminus S \mid T \in \mathcal S, T \setminus S \neq \varnothing\})$.
\item \label{itm:sc3}Remove all sets $|S| > d$ from $\mathcal S$.
Since now all sets have size at most $d$, each band $B_i$ must contain a partial cover of any minimum set cover.
\item \label{itm:sc4}Calculate the size of the minimum set cover $f(P,\mathcal S)$ for all sets $|P| \leq (1-\alpha) n/4+d/2$ using classical dynamic programming.
\item \label{itm:sc5}Run quantum minimum finding over subsets $P \in B_3$ to find the answer,
$$f(\mathcal U, \mathcal S) = \min_{P \in B_3}\{ f(P,\mathcal S) + f(\mathcal U\setminus P,\mathcal S)\}.$$
Note that $\mathcal U \setminus P \in B_3$.
To calculate $f(P,\mathcal S)$ for a set $P \in B_i$ such that $i \in \{2, 3\}$, run quantum minimum finding over the subsets of $P$ in $B_{i-1}$ to find $$f(P,\mathcal S) = \min_{Q \in B_{i-1} \atop Q \subset P}\{ f(Q,\mathcal S) + f(P \setminus Q,\mathcal S)\}.$$
If $||P| - \alpha n/4| \leq d/2$ (that is, if $P \in B_1$) or $||P| - (1-\alpha) n/4| \leq d/2$, then $f(P,\mathcal S)$ is known from the preprocessing Step \ref{itm:sc4}.
\item Pick the best solution from Steps \ref{itm:sc2} and \ref{itm:sc5} and return the answer.
\end{enumerate}
\caption{Quantum algorithm for \textsc{Set Cover}}
\end{algorithm}

\subsection{Running Time}

First, we estimate the maximum depth of the recursion $t$.
If the algorithm is called recursively at Step \ref{itm:sc2}, the number of elements in the universe decreases by $d=\beta n$.
Hence we must have $(1-\beta)^t N < c N$, which implies that $t < \log_{1-\beta} c = O(1)$.
Hence for a fixed $\beta$, the depth of the recursion is constant.

At each recursive call, there is a $\poly(m,n)$ cost of iterating over subsets from $\mathcal S$ or iterating over elements of some $S \in \mathcal S$.
Therefore, there is only a $\poly(m,n)^t = \poly(m,n)$ additional subexponential multiplicative cost.

If the algorithm terminates at Step \ref{itm:sc1}, then it runs in $2^{cn} = \delta^n$ exponential time.
Recall that $\delta^n$ is also the exponential complexity of the main quantum subroutine.
We now estimate $\delta$.
The classical preprocessing in Step \ref{itm:sc4} requires exponential time
$${n \choose \leq (1-\alpha) n/4+d/2} \leq 2^{\be\left(\frac{1-\alpha}{4}+\frac{\beta}{2}\right)n}.$$
The quantum subroutine in Step \ref{itm:sc5} requires exponential time at most
\begin{align*}
& \sqrt{\left(\sum_{i=-d/2}^{d/2} {n \choose n/2+i}\right) \left(\sum_{i=-d/2}^{d/2} {n/2+d/2 \choose n/4+i}\right) \left(\sum_{i=-d/2}^{d/2} {n/4+d/2 \choose \alpha n/4 + i}\right)} \\
 \leq~& \sqrt{{n \choose \leq n/2+d/2} {n/2+d/2 \choose \leq n/4+d/2}{n/4+d/2 \choose \leq \alpha n/4 + d/2}}\\
 \leq~& 2^{\frac 1 2 \left( 1 +  \frac{1+\beta}{2} + \be\left(\frac{\alpha+2\beta}{1+2\beta}\right)\frac{1+2\beta}{4} \right)n}.
\end{align*}

If we set $\beta$ to be negligibly small and then balance the preprocessing and quantum running times, we obtain the same solution for $\alpha$ as for the \textbf{TravellingSalesman} algorithm in Section \ref{sec:tsptime}, $\alpha = 0.055362$.
The running time of the algorithm then is $O(\poly(m,n)2^{0.788595n}) = O(\poly(m,n)1.727391...^n)$.

\printbibliography

\end{document}